\documentclass[11pt]{article}

\usepackage[margin=1in]{geometry}
\usepackage{amsmath,amsthm,amsfonts,amssymb,graphicx}
\usepackage{thm-restate}

\usepackage[all]{xy}

\usepackage[pagebackref]{hyperref} 
\usepackage{hypcap}
\hypersetup{
    bookmarksnumbered=true, 
    unicode=false, 
    pdfstartview={FitH}, 
    pdftitle={Title}, 
    pdfauthor={Authors}, 
    pdfsubject={}, 
    pdfcreator={}, 
    pdfproducer={}, 
    pdfkeywords={}, 
    pdfnewwindow=true, 
    colorlinks=true, 
    linkcolor=blue, 
    citecolor=blue, 
    filecolor=blue, 
    urlcolor=blue 
}

\usepackage{cite}

\usepackage{tikz}

\renewcommand{\backref}[1]{}

\renewcommand{\backrefalt}[4]{%
\ifcase #1 %
\or 
[p.\ #2]%
\else 
[pp.\ #2]%
\fi}

\newtheorem{proposition}{Proposition}
\newtheorem{corollary}{Corollary}

\newtheorem{question}{Question}

\theoremstyle{definition}
\newtheorem{definition}{Definition}

\newcommand{\eq}[1]{\hyperref[eq:#1]{(\ref*{eq:#1})}}
\renewcommand{\sec}[1]{\hyperref[sec:#1]{Section~\ref*{sec:#1}}}
\newcommand{\thm}[1]{\hyperref[thm:#1]{Theorem~\ref*{thm:#1}}}
\newcommand{\lem}[1]{\hyperref[lem:#1]{Lemma~\ref*{lem:#1}}}
\newcommand{\prop}[1]{\hyperref[prop:#1]{Proposition~\ref*{prop:#1}}}
\newcommand{\cor}[1]{\hyperref[cor:#1]{Corollary~\ref*{cor:#1}}}
\newcommand{\fig}[1]{\hyperref[fig:#1]{Figure~\ref*{fig:#1}}}
\newcommand{\tab}[1]{\hyperref[tab:#1]{Table~\ref*{tab:#1}}}
\newcommand{\app}[1]{\hyperref[app:#1]{Appendix~\ref*{app:#1}}}
\newcommand{\ques}[1]{\hyperref[q:#1]{Question~\ref*{q:#1}}}

\newcommand{\comment}[1]{}  

\newcommand{\D}{D^\mathrm{sc}}
\newcommand{\R}{R^\mathrm{sc}}
\newcommand{\Rz}{R_0^\mathrm{sc}}
\newcommand{\e}{\mathrm{exp}}

\newcommand{\defeq}{=}

\newcommand{\MAJ}{\mathsf{MAJ}}
\newcommand{\FMAJ}{\mathsf{4\textrm{-}MAJ}}
\newcommand{\rT}{\mathrm{T}}

\newcommand{\PPRT}{\mathrm{PPRT}}

\begin{document}


\title{Separating decision tree complexity\\ from subcube partition complexity}

\author{
\normalsize Robin Kothari\footnote{Center for Theoretical Physics, Massachusetts Institute of Technology. Part of this work was performed while the author was a student at the David R.\ Cheriton School of Computer Science and the Institute for Quantum Computing, University of Waterloo.
\texttt{rkothari@mit.edu}
} 
\and
\normalsize David Racicot-Desloges\footnote{D\'epartement de Math\'ematiques, Facult\'e des Sciences, Universit\'e de Sherbrooke.
Part of this work was completed while the author was an intern student at the Centre for Quantum Technologies, National University of Singapore.
\texttt{David.Racicot-Desloges@USherbrooke.ca}
} 
\and
\normalsize Miklos Santha\footnote{LIAFA, CNRS, Universit\'e Paris Diderot, Paris, France and 
Centre for Quantum Technologies, National University of Singapore.
\texttt{miklos.santha@gmail.com}
}
}

\date{}
\maketitle


\begin{abstract}
The subcube partition model of computation is at least as powerful as decision trees 
but no separation between these models was known. We show that there exists a function whose deterministic
subcube partition complexity is asymptotically smaller than its 
randomized decision tree complexity, resolving an open problem of Friedgut, Kahn, and Wigderson (2002). Our lower bound is based on the information-theoretic techniques first introduced to lower bound the randomized decision tree complexity of the recursive majority function.

We also show that the public-coin partition bound, 
the best known lower bound method for randomized decision tree complexity subsuming other general techniques such as block sensitivity, approximate degree, randomized certificate complexity, and the classical adversary bound, also lower bounds randomized subcube partition complexity. This shows that all these lower bound techniques cannot prove optimal lower bounds for randomized decision tree complexity, which answers an open question of Jain and Klauck (2010) and Jain, Lee, and Vishnoi (2014).

\end{abstract}

\section{Introduction}
\label{sec:intro}

The decision tree is a widely studied model of computation. While we have made significant progress in understanding this model (e.g., see the survey by Buhrman and de Wolf \cite{BW02}), questions from over 40 years ago still remain unsolved \cite{Ros73}.

In the decision tree model, we wish to compute a function $f:\{0,1\}^n \to \{0,1\}$ on an input $x \in \{0,1\}^n$, but we only have access to the input via a black box. The black box can be queried with an index $i \in [n]$, where $[n] \defeq  \{1,2,\ldots,n\}$, and will respond with the value of $x_i$, the $i$th bit of $x$. The goal is to compute $f(x)$, while minimizing the number of queries made to the black box. 

For a function  $f:\{0,1\}^n \to \{0,1\}$, let $D(f)$ denote the deterministic query complexity (or decision tree complexity) of computing $f$, the minimum number of queries made by a deterministic algorithm that computes $f$ correctly on all inputs. Let $R_0(f)$ denote the zero-error randomized query complexity of computing $f$, the minimum expected cost of a zero-error randomized algorithm
that computes $f$ correctly on all inputs. Finally, let $R(f)$ denote the bounded-error randomized query complexity of computing $f$, the number of queries made in the worst case by a randomized algorithm that outputs $f(x)$ on input $x$ with probability at least 2/3. More precise definitions can be found in \sec{prelim}.

Several lower bound techniques have been developed for query complexity over the years, most of which are based on the following observation: A decision tree that computes $f$ and makes $d$ queries partitions the set of all inputs, the hypercube $\{0,1\}^n$, into a set of monochromatic subcubes where each subcube has at most $d$ fixed variables. A subcube is a restriction of the hypercube in which the values of some subset of the variables have been fixed. For example, the set of $n$-bit strings in which the first variable is set to 0 is a subcube of $\{0,1\}^n$ with one fixed variable. 
A subcube is monochromatic if $f$ takes the same value on all inputs in the subcube. This idea is also the basis of many lower bound techniques in communication complexity \cite{KN06}, where a valid protocol partitions the space of inputs into monochromatic rectangles.

However, not all subcube partitions arise from decision trees, which naturally leads to a potentially more powerful model of computation. 
This model is 
called the subcube partition model in \cite{FKW02}, but has been studied before under different names (see e.g., \cite{BOH90}). The deterministic subcube partition complexity of $f$, denoted by $\D(f)$, is the minimum $d$ such that there is a partition of the hypercube into a set of monochromatic subcubes in which each subcube has at most $d$ fixed variables. Since a decision tree making $d$ queries always gives rise to such a partition, we have $\D(f)\leq D(f)$. Similarly, we define zero-error and bounded-error versions of subcube partition complexity, denoted by $\Rz(f)$ and $\R(f)$, respectively, and obtain the inequalities  $\Rz(f) \leq R_0(f)$ and $\R(f) \leq R(f)$. As expected, we also have $\Rz(f) \leq \D(f)$
and $\R(f) \leq \D(f)$.

This brings up the obvious question of whether these models are equivalent. Separating 
them is difficult, precisely because most lower bound techniques for query complexity also lower bound subcube partition complexity. 
The analogous question in communication complexity is also a long-standing open problem (see \cite[Open Problem 2.10]{KN06} or \cite[Chapter 3.2]{Juk12}). In fact, Friedgut, Kahn, and Wigderson \cite[Question 1.1]{FKW02} explicitly ask if these measures are asymptotically different in the randomized model with zero error:

\begin{question}\label{q:fkw}
Is there a function (family) $f= (f_n)$ such that $\Rz(f) = o(R_0(f))$?   
\end{question}

Similarly, one can ask the same question for bounded-error randomized query complexity. The main result of this paper resolves these questions: 

\begin{restatable}{theorem}{main}
\label{thm:main}
There exists a function $f= (f_h)$, with $f_h:\{0,1\}^{4^h} \to \{0,1\}$, such that $\D(f) \leq 3^h$, but $D(f) = 4^h$, $R_0(f) \geq 3.2^h$, and $R(f) = \Omega(3.2^h)$.
\end{restatable}

This shows that query complexity and subcube partition complexity are asymptotically different in the deterministic, zero-error, and bounded-error settings. Besides resolving this question, our result has another application.
We know several techniques to lower bound bounded-error randomized query complexity, such as approximate polynomial degree \cite{NS94}, block sensitivity \cite{N91}, randomized certificate complexity \cite{A06} and the classical adversary bound \cite{LM08,SS06,A08}. All these techniques are subsumed by the partition bound of Jain and Klauck \cite{JK10}, which in turn is subsumed by the public-coin partition bound  of Jain, Lee, and Vishnoi \cite{JLV14}. Additionally, this new lower bound is within a quadratic factor of randomized query complexity. In other words, if $\PPRT(f)$ denotes the bounded-error public-coin partition bound for a function $f$, we have $\PPRT(f) \leq R(f)$ and also $R(f) = O(\PPRT(f)^2)$. This leaves open the intriguing possibility that this technique is optimal and is asymptotically equal to bounded-error randomized query complexity. Jain, Lee, and Vishnoi \cite{JLV14} 
indeed ask the following question:

\begin{question}\label{q:jain}
Is there a function (family) $f= (f_n)$ such that $\PPRT(f) =  o(R(f))$?
\end{question}

Our result also answers this question, because, as we show in \sec{prelim}, $\PPRT(f) \leq \R(f)$. Thus, our asymptotic separation between $\R(f)$ and $R(f)$ also separates $\PPRT(f)$ from $R(f)$.

We now provide a high-level overview of the techniques used in this paper. The main result is based on establishing the various complexities of a certain function. The function we choose is based on the quarternary majority function $\FMAJ:\{0,1\}^4 \to \{0,1\}$, defined as the majority of the four input bits, with ties broken by the first bit. 
This function has low deterministic subcube complexity, $\D(\FMAJ) \leq 3$, but has deterministic query complexity $D(\FMAJ) = 4$. From this function, we define an iterated function $\FMAJ_h$ on $4^h$ variables by composing the function with itself $h$ times, which gives us a function on $4^h$ bits. Since deterministic query complexity and deterministic subcube complexity behave nicely under composition, we have $D(\FMAJ_h) = 4^h$ and $\D(\FMAJ_h) \leq 3^h$. These results are further discussed in \sec{it-fmaj}. 
To prove \thm{main}, it remains to show that the randomized query complexity of this function is $\Omega(3.2^h)$.

We lower bound the randomized query complexity of $\FMAJ_h$ using a strategy similar to the information-theoretic technique of Jayram, Kumar, and Sivakumar \cite{JKS03} and its simplification by Landau, Nachmias, Peres, and Vanniasegaram \cite{LNPV06}. However, the original strategy was applied to lower bound a symmetric function (iterated 3-$\MAJ$), whereas our function is not symmetric since the first variable of $\FMAJ$ is different from the rest. We modify the technique to apply it to asymmetric functions and establish the claimed lower bound. The lower bound relies on choosing a ``hard distribution'' of inputs and establishing a recurrence relation between the complexities of the function and its subfunctions on this distribution. 
Unlike 3-$\MAJ$, where there is a natural candidate for a hard distribution, our chosen distribution is not obvious and is constrained by the fact that it must fit nicely into these recurrence relations. We prove this lower bound in \sec{lb-fmaj}.
We end with some discussion and open problems in \sec{disc}.

\section{Preliminaries}
\label{sec:prelim}

In this section, we formally define the various models of query complexity and subcube partition complexity, and the partition bound \cite{JK10} and public-coin partition bound \cite{JLV14}. We then study the relationships between these quantities.

For the remainder of the paper, let $f : \{0,1\}^n \rightarrow \{0,1\}$ be a Boolean function on $n$ bits and $x=(x_1,x_2,\dots,x_n) \in \{0,1\}^n$ be any input. Let $[n]$ denote the set $\{1,2,\ldots,n\}$ and let the support of a probability distribution $p$ be denoted by $\text{supp}(p)$. 
Lastly, we require the notion of composing two Boolean functions. If $f : \{0,1\}^n \rightarrow \{0,1\}$  and
$g : \{0,1\}^m \rightarrow \{0,1\}$ are two Boolean functions,  the {\em composed function} $f \circ g : \{0,1\}^{nm} \rightarrow \{0,1\}$ 
acts on the Boolean string $y = (y_{11}, \ldots, y_{1m}, y_{21}, \ldots, y_{nm})$ as
$f \circ g(y) = f ( g (y_{11}, \ldots, y_{1m}), \ldots , g (y_{n1}, \ldots, y_{nm})).$

\subsection{Decision tree or query complexity}

The deterministic query complexity of a function $f$, $D(f)$, is the minimum number of queries made by a deterministic algorithm that computes $f$ correctly.

Formally, a \emph{deterministic decision tree $A$ on $n$ variables} is a binary tree in which each leaf is labeled by either a $0$ or a $1$, and each internal node is labeled with a value $i \in [n]$. 
For every internal node of $A$, one of the two outgoing edges is labeled $0$ and the other edge is labeled $1$. On an input $x$, the algorithm $A$ follows the unique path from the root to one of its leaves in the natural way: for an internal node labeled with the value $i$, it follows the outgoing edge labeled by $x_i$. 
The output $A(x)$ of the algorithm $A$ on input $x$ is the label of the leaf of this path. We say that the decision tree $A$ \emph{computes} $f$ if $A(x)=f(x)$ for all $x$.

We define the {\em cost} of algorithm $A$ on input $x$, denoted by $C(A,x)$, to be the number of bits queried by $A$ on $x$, that is the number of internal nodes evaluated by $A$ on $x$.  The cost of an algorithm $A$, denoted $C(A)$, is the worst-case cost of the algorithm over all inputs $x$, that is $C(A) \defeq \max_{x}C(A,x)$. Now, let $\mathcal{D}_n$ denote the set of all deterministic decision trees on $n$ variables and let $\mathcal{D}(f) \subseteq \mathcal{D}_n$ be the set of all deterministic decision trees that compute $f$. 
We define the \emph{deterministic query complexity of $f$} as $D(f) \defeq \min_{A \in \mathcal{D}(f)}C(A)$.

One of the features of deterministic query complexity that we use in this paper is its composition property \cite{Mon14}. This property is very intuitive: it asserts that the best way to compute the composition of $f$ and $g$ is to use optimal algorithms for $f$ and $g$ independently.
\begin{proposition}
\label{prop:det-comp}
For any two Boolean functions $f$ and $g$, $D(f\circ g)=D(f)D(g)$.
\end{proposition}

We can now move on to randomized analogs of deterministic query complexity. 
In a randomized algorithm, the choice of the queries might also depend on some randomness. 
Formally, a \emph{randomized decision tree $B$ on $n$ variables} is defined by a probability distribution $b$ 
over $\mathcal{D}_n$, that is by a function $b:\mathcal{D}_n \to [0,1]$ such that $\sum_{A \in \mathcal{D}_n} b(A) = 1$.
On an input $x$, the algorithm $B$ picks a deterministic decision tree $A$ with probability $b(A)$ and outputs $A(x)$. Thus, for every $x$, the value $B(x)$ of $B$ on $x$ is a random variable. 

We say that a randomized algorithm $B$ computes $f$ with 
{\it error} $\varepsilon \geq 0$ if $\Pr[B(x)=g(x)] \geq 1-\varepsilon$ for all $x$, that is if $\sum_{A(x)=f(x)} b(A) \geq  1-\varepsilon$ for all $x$.
Let $\mathcal{R}_n$ be the set of all randomized decision trees over $n$ bits and let
$\mathcal{R}_{\varepsilon}(f) \subseteq \mathcal{R}_n$ be the set of all randomized decision trees that compute $f$ with error $\varepsilon$. A randomized algorithm $B$ then computes $f$ with zero error if $\text{supp}(b) \subseteq \mathcal{D}(f)$,
that is the probability distribution $b$ is completely supported on  the set of deterministic decision trees that compute $f$. A zero-error randomized algorithm, also known as a Las Vegas algorithm, always outputs the correct answer. The cost of a
zero-error randomized algorithm $B$ on $x$ is defined as $C(B,x) \defeq \sum_{A \in \mathcal{D}_n} b(A) C(A,x) = \mathbb{E}[C(A,x)]$, the expected number of queries made on input $x$. 
The {\em zero-error randomized query complexity of $f$}, denoted by $R_0(f)$, is defined as $R_0(f)= \min_{B \in \mathcal{R}_{0}(f)}\max_{x} C(B,x)$. 
From the definition of zero-error randomized query complexity, it is clear that $R_0(f) \leq D(f)$. 
The complexity $R_0(f)$ can be of strictly smaller order of growth than $D(f)$:
there exists a function $f$ for which $R_0(f) = o( D(f))$, e.g., the iterated \textsf{NAND}-function~\cite{SW86}.

Randomized algorithms with error $\varepsilon > 0$ might give incorrect answer on their inputs with probability $\varepsilon$.
We say that a randomized algorithm is
of {\em bounded-error} (sometimes called a Monte Carlo algorithm) if on any input $x$, the probabilistic output is incorrect with probability at most $1/3$. The constant $1/3$ is not important and replacing it with any constant strictly between $0$ and $1/2$ will only change the complexity by a constant multiplicative factor. 
For $\varepsilon>0$, 
the cost of an
$\varepsilon$-error randomized algorithm $B$ on $x$ is defined as $C(B,x) = \max_{A \in \text{supp}(b) } C(A,x)$, 
the maximum number of queries made on input $x$ by an algorithm in the support of $b$. 
Note how this definition differs from the one given for the zero-error case.
We define the \emph{$\varepsilon$-error randomized complexity of $f$} as $R_\varepsilon(f)\defeq \min_{B \in \mathcal{R}_{\varepsilon}(f)}\max_{x}C(B,x)$, and the {\em bounded-error randomized query complexity of $f$} as $R(f) \defeq R_{1/3}(f)$.
Note that this definition is valid only for $\varepsilon>0$ and does not coincide with $R_0(f)$ defined above for $\varepsilon=0$. Setting $\varepsilon=0$ in this definition simply gives us the deterministic query complexity $D(f)$. 
Nonetheless, it is true that $R(f) = O(R_0(f))$. This distinction is discussed in more detail in \sec{disc}. Lastly, note that for all $\varepsilon>0$, we have $R_\varepsilon(f) \leq D(f)$, and that there exist functions for which $R(f) = o(D(f))$~\cite{SW86}.


In order to establish lower bounds on randomized query complexity, it is useful to take a distributional view of randomized algorithms
\cite{Yao77}, that is to consider the performance of randomized algorithms on a chosen distribution over inputs. Let $\mu$ be a probability distribution over all possible inputs of length $n$, and let $B$ be a randomized decision tree algorithm.
The cost of $B$ under $\mu$ is $C(B, \mu) = \sum_{x\in \{0,1\}^n}   \mu(x)  C(A,x) = \mathbb{E}[C(B,x)]$. 
We define the 
\emph{$\varepsilon$-error distributional complexity of $f$ under $\mu$} as 
$\Delta_{\varepsilon}^{\mu}(f) \defeq \min_{B \in \mathcal{R}_{\varepsilon}(f)} C(B, \mu)$.
The following simple fact is the basis of many lower bound arguments.

\begin{proposition}
\label{prop:yao}
For every distribution $\mu$ over $\{0,1\}^n$, 
and for all $\varepsilon \geq 0,$
we have $\Delta_{\varepsilon}^{\mu}(f) \leq R_\varepsilon(f)$.
\end{proposition}

\begin{proof}
This follows by expanding out the definitions and using the simple inequality between expectation and maximum: 
\begin{equation}
\Delta_{\varepsilon}^{\mu}(f)  =  \min_{B \in \mathcal{R}_{\varepsilon}(f)} C(B, \mu) 
 =  \min_{B \in \mathcal{R}_{\varepsilon}(f)}  \mathbb{E}[C(B,x)]
 \leq \min_{B \in \mathcal{R}_{\varepsilon}(f)}  \max_x C(B,x)
  =  R_\varepsilon(f). \qedhere 
\end{equation}
\end{proof}

\subsection{Subcube partition complexity}

A subcube of the hypercube $\{0,1\}^n$ is a set of $n$-bit strings obtained by fixing the values of some subset of the variables. In other words, a subcube is the set of all inputs consistent with a partial assignment of $n$ bits. Formally, a \emph{partial assignment on $n$ variables} is a function $a:I_a \rightarrow \{0,1\}$, with $I_a \subseteq [n]$.
Given a partial assignment $a$, we call $S(a)=\{y \in \{0,1\}^n : y_i=a(i) \text{ for all } i \in I_a\}$ the \emph{subcube generated by $a$}. A set $S \subseteq \{0,1\}^n$ is a \emph{subcube} 
of the hypercube $\{0,1\}^n$ if $S=S(a)$ for some partial assignment on $n$ variables $a$. Clearly, for every subcube, there exists exactly one such $a$. We denote by $I_S$ the domain $I_a \subseteq [n]$ of $a$ where $S=S(a)$. 
%
For example, the set $\{0100, 0101, 0110, 0111\}$ is a subcube of $\{0,1\}^4$. It is generated by the partial assignment $a:\{1,2\} \to \{0,1\}$, where $a(1)=0$ and $a(2)=1$. An alternative representation of a partial assignment is by an $n$-bit string where a position $i$ takes the value $a(i)$ if $i \in I_a$ and takes the value $*$ otherwise. For this example, the subcube  $\{0100, 0101, 0110, 0111\}$ is generated by the partial assignment $01**$. Finally, another useful representation is in terms of a conjunction of literals, that is satisfied by all strings in the subcube. For example, the subcube $\{0100, 0101, 0110, 0111\}$ consists exactly of all $4$-bit strings that satisfy the formula $\overline{x_1} \wedge {x_2}$.

The subcube partition model of computation, studied  previously in \cite{FKW02,BOH90,CKLS13}, is a generalization of the decision tree model.
 A {\em partition}  $\{S_1, \ldots, S_{\ell}\}$ of $\{0,1\}^n$ is a set of pairwise disjoint subsets of $\{0,1\}^n$ that together cover the entire hypercube, that is $\bigcup_i S_i = \{0,1\}^n$ and  $S_i \cap S_j = \emptyset$ for $i\neq j$.

A \emph{deterministic subcube partition $P$ on $n$ variables} is a partition of $\{0,1\}^n$ with a Boolean value $s \in \{0,1\}$ associated to each subcube, that is $P=\{(S_1,s_1),(S_2,s_2),\dots,(S_{\ell},s_{\ell})\}$, where each $S_i$ is a subcube and $\{S_1, \ldots, S_{\ell}\}$ is a partition of $\{0,1\}^n$.
If the assignment $a$ generates $S_i$ for some $i$, we call $a$ a {\em generating assignment for $P$}. For any $x$, we let $S^x$ denote the subcube containing $x$, that is, if $x \in S_i$, then $S^x  \defeq S_i$. We  define the value $P(x)$ of $P$ on $x$ as $s_i$. 

We say that a deterministic subcube partition $P$ computes $f$ if $P(x)=f(x)$ for all $x$.
Note that every deterministic decision tree algorithm $A$ computing $f$ induces a subcube partition computing $f$ that consists of the subcubes generated by the partial assignments defined by the root--leaf paths of the tree and the Boolean values of the corresponding leaves. 
We define the cost of $P$ on $x$ as $C(P,x) \defeq |I_{S^x}|$, analogous to the number of queries made on input $x$ in query complexity. We define the worst-case cost as $C(P) \defeq  \max_x C(P,x)$. Let  $\mathcal{D}^\mathrm{sc}_n$ be the set of all deterministic subcube partitions on $n$ variables and let $\mathcal{D}^\mathrm{sc}(f) \subseteq \mathcal{D}^\mathrm{sc}_n$ be those partitions that compute $f$. We define the \emph{deterministic subcube partition complexity of $f$} as $\D(f)= \min_{P \in \mathcal{D}^\mathrm{sc}(f)} C(P)$. Deterministic subcube partition complexity also satisfies a composition theorem.

\begin{proposition}
\label{prop:dsc-comp}
For any $f : \{0,1\}^n \rightarrow \{0,1\}$  and $g : \{0,1\}^m \rightarrow \{0,1\}$, $\D(f\circ g) \leq \D(f)\D(g)$.
\end{proposition}

\begin{proof}
Let $P = \{(S_1,s_1),(S_2,s_2),\dots,(S_p,s_p)\}$ and $Q = \{(T_1,t_1),\dots,(T_q,t_q)\}$ be optimal deterministic subcube partitions computing $f$ and $g$ respectively. Suppose that $S_h$ is generated by $a_h$ for $h \in [p]$, and that $T_j$ is generated by $b_j$ for $j \in [q]$.
Let $I_{a_h} = \{ i_1, \ldots , i_{c_h}\}$. 
We define the deterministic subcube partition $P \circ Q$ on $nm$ variables as follows. The generating assignments for $P \circ Q$
are $a_h \circ (b_{j_1}, \ldots, b_{j_{c_h}})$, for all $h \in [p]$, and  $ j_1, \ldots, j_{c_h} \in [q]$ that satisfy 
$a(i_k) = t_{j_k}$ for $k \in [c_h]$.
When $|I_{b_{j_k}}| = d_k$, the assignment 
$e = a_h \circ (b_{j_1}, \ldots, b_{j_{c_h}})$ is defined by $I_e = \{ (1,1), \ldots, (1, d_1), (2,1), \ldots , (c_h, d_{c_h})\}$, and $e(k,r) = b_{j_k}(r)$
for $ 1 \leq r \leq d_{k}.$
The Boolean value associated with $e$ is $s_h$. It is easy to check that $P \circ Q$ computes $f \circ g$ and that
$C(P \circ Q) \leq C(P) C(Q)$.
\end{proof}

As in the case of query complexity, we extend deterministic subcube complexity to the randomized setting.
A \emph{randomized subcube partition $R$ on $n$ variables} is given by a distribution 
$r$ over all deterministic subcube partitions on $n$ variables. 
As for randomized decision trees, $R(x)$ is a random variable and we say that $R$ computes $f$ with error $\varepsilon \geq 0$ if $\Pr[R(x)=f(x)] \geq 1-\varepsilon$ for all $x$. Let $\mathcal{R}^\mathrm{sc}_n$ be the set of all randomized subcube partitions over $n$ variables  and $\mathcal{R}^\mathrm{sc}_{\varepsilon}(f) \subseteq \mathcal{R}^\mathrm{sc}_n$ be the set of all randomized subcube partitions that compute $f$ with error $\varepsilon$. 

The cost of a zero-error randomized subcube partition $R$ on $x$ is defined by $C(R,x) = \mathbb{E}[C(P,x)]$,
where the expectation is taken over $R$. For an $\varepsilon$-error subsucbe  partition $R$, with $\varepsilon>0$, the 
cost on $x$ is $C(R,x) = \max_{P \in \text{supp}(r) }C(P,x)$.
For $\varepsilon \geq 0$, we define the \emph{$\varepsilon$-error randomized subcube complexity of $f$} by 
$\R_\varepsilon(f)= \min_{R \in \mathcal{R}^\mathrm{sc}_{\varepsilon}(f)}  \max_{x} C(R,x) $. 

As mentioned before, a deterministic decision tree induces a deterministic subcube partition with the same cost and thus a randomized decision tree induces a randomized subcube partition with the same cost, which yields the following.

\begin{proposition}
\label{prop:query-sc}
For an $n$-bit Boolean function $f:\{0,1\}^n \to \{0,1\}$, we have that $\D(f) \leq D(f)$ and, for all $\varepsilon \geq 0$, we have that $\R_\varepsilon(f) \leq R_\varepsilon(f)$.
\end{proposition}

\subsection{Partition bounds}
In 2010, Jain and Klauck \cite{JK10} introduced a linear programming based lower bound technique for randomized query complexity called the partition bound. They showed that it subsumes all known general lower bound methods for randomized query complexity, including approximate polynomial degree \cite{NS94}, block sensitivity \cite{N91}, randomized certificate complexity \cite{A06}, and the classical adversary bound \cite{LM08,SS06,A08}. 

Recently, Jain, Lee, and Vishnoi \cite{JLV14} presented a modification of this method called the public-coin partition bound, which is easily seen to be stronger than the partition bound. Furthermore, they were able to show that the gap between this new lower bound and randomized query complexity can be at most quadratic. 
We define these lower bounds formally.

\begin{definition}[Partition bound]
\label{def:prt}
Let $f:\{0,1\}^n \to \{0,1\}$ be an $n$-bit Boolean function and let $\mathcal{S}_n$ denote the set of all subcubes of $\{0,1\}^n$. Then, for any $\varepsilon \geq 0$, let $\text{prt}_\varepsilon(f)$ be the optimal value of the following linear program:
\begin{align}
\mathop{\textrm{minimize: }}_{w_{S,z}} & \sum_{z=0}^{1} \sum_{S \in \mathcal{S}_n} w_{S,z} \cdot 2^{|I_S|} \\
\text{subject to: } 
&\sum_{S:x \in S} w_{S,f(x)} \geq 1-\varepsilon \qquad (\textrm{for all }x \in \{0,1\}^n),  \\
&\sum_{S:x \in S} \sum_{z=0}^1 w_{S,z} = 1  \qquad (\textrm{for all }x \in \{0,1\}^n),  \\
&w_{S,z} \geq 0 \qquad \textrm{(for all }S \in \mathcal{S}_n \text{ and } z \in \{0,1\}).
\end{align}
The  {\em $\varepsilon$-partition bound} of $f$ is defined as $\text{PRT}_\varepsilon(f) \defeq \frac{1}{2}\log_2(\text{prt}_\varepsilon(f))$.
\end{definition}

We now define the public-coin partition bound. Although our definition differs from the original definition \cite{JLV14}, it is not too hard to see that they are equivalent. Before presenting the definition, recall that $\mathcal{D}^\textrm{sc}_n$ is the set of deterministic subcube partitions on $n$ variables, and $\mathcal{R}^\textrm{sc}_\varepsilon(f)$ is the set of randomized subcube partitions that compute $f$ with error at most $\varepsilon \geq 0$. For a randomized subcube partition $R \in \mathcal{R}^\textrm{sc}_\varepsilon(f)$, we let $r$ be the probability distribution over deterministic subcube partitions corresponding to $R$.

\begin{definition}[Public-coin partition bound]
\label{def:pprt}
Let $f:\{0,1\}^n \to \{0,1\}$ be an $n$-bit Boolean function. Then, for any $\varepsilon \geq 0$, let $\text{pprt}_\varepsilon(f)$ be the optimal value of the following linear program:
\begin{align}
\mathop{\textrm{minimize: }}_{R}  
& \sum_{z=0}^1 \sum_{S \in \mathcal{S}_n} \sum_{P:(S,z) \in P} r(P) \cdot 2^{|I_S|} \\
\textrm{subject to: } & R \in \mathcal{R}^\textrm{sc}_{\varepsilon}(f).
\end{align}
The \emph{$\varepsilon$-public-coin partition bound of $f$} is defined as $\text{PPRT}_\varepsilon(f) \defeq \frac{1}{2}\log_2(\text{pprt}_\varepsilon(f))$.
\end{definition}

Using the original definition, it is trivial that $\text{prt}_\varepsilon(f) \leq \text{pprt}_\varepsilon(f)$, since the public-coin partition bound is defined using the same linear program, with additional constraints. This statement also holds with the definitions given above, as we now prove.

\begin{restatable}{proposition}{prt}
For any 
Boolean function $f$ and for all $\varepsilon \geq 0$, we have that $\mathrm{prt}_\varepsilon(f) \leq \mathrm{pprt}_\varepsilon(f)$.
\end{restatable}

\begin{proof}
Let $R'$ be a randomized subcube partition achieving the optimal value for the linear program of $\text{pprt}_\varepsilon(f)$ and $r'$ be the corresponding probability distribution over deterministic subcube partitions. Then, for all $(S,z)$ where $S$ is a subcube and $z\in \{0,1\}$, let 
\begin{equation}
w'_{S,z}=\displaystyle \sum_{P:(S,z) \in P} r'(P).
\end{equation}
This family of variables satisfies the conditions of the $\text{pprt}_\varepsilon(f)$ linear program and is such that
\begin{equation}
\sum_{z=0}^{1} \sum_{S \in \mathcal{S}_n} w'_{S,z} \cdot 2^{|I_S|} = 
 \sum_{z=0}^1 \sum_{S \in \mathcal{S}_n} \sum_{P:(S,z) \in P} r'(P) \cdot 2^{|I_S|}\thinspace.
\end{equation}
\end{proof}

Recall that both partition bounds lower bound randomized query complexity, as shown in \cite{JLV14}. In particular, for all $\varepsilon > 0$, $\text{PRT}_\varepsilon(f) \leq \text{PPRT}_\varepsilon(f) \leq R_\varepsilon(f)$ and, when $\varepsilon=0$, we have that $\text{PRT}_0(f) \leq \text{PPRT}_0(f) \leq D(f)$. It is not known if the zero-error partition bound also lower bounds zero-error randomized query complexity. However, as mentioned, the partition bounds also lower bound subcube partition complexity, which implies that they lower bound query complexity. The proof for query complexity easily extends to subcube partition complexity.

\begin{restatable}{proposition}{pprtrsc}
\label{prop:pprt-rsc}
For every Boolean function $f$ and for all $\varepsilon > 0$, we have that
$\mathrm{PPRT}_\varepsilon(f)   \leq \R_\varepsilon(f)$ and $\mathrm{PPRT}_0(f) \leq \D(f)$.
\end{restatable}

\begin{proof}
Let $R' \in \mathcal{R}^\mathrm{sc}_\varepsilon(f)$ be a randomized subcube partition that achieves $\R_\varepsilon(f)$ and let $r'$ be its corresponding probability distribution over deterministic subcube partitions. Let $P \in \text{supp}(r')$. 
By definition, for every $(S,z) \in P$, we have that $|I_S|  \leq  C(P)$. Also by definition, $C(P) \leq   \R_\varepsilon(f)$.
Furthermore, if $P=\{(S_1,z_1),(S_2,z_2),\dots,(S_m,z_m)\}$, then 
\begin{equation}
|P| \cdot 2^{n-C(P)} = m \cdot 2^{n-C(P)} \leq \displaystyle \sum_{i=1}^m 2^{n-|I_{S_i}|} = 2^n.
\end{equation}
This implies that $|P| \leq 2^{C(P)} \leq 2^{\R_\varepsilon(f)}$ and, therefore, that
\begin{align}
\text{pprt}_\varepsilon(f) & = \sum_{z=0}^{1} \sum_{S \in \mathcal{S}_n} \sum_{P:(S,z) \in P} r'(P) \cdot 2^{|I_S|} 
\ \leq \  2^{\R_\varepsilon(f)}\sum_{z=0}^{1} \sum_{S \in \mathcal{S}_n} \sum_{P:(S,z) \in P} r'(P) \\
& = 2^{\R_\varepsilon(f)} \sum_{P \in \text{supp}(r')} r'(P) \cdot |P| 
\ \leq \ 2^{\R_\varepsilon(f)} \cdot 2^{\R_\varepsilon(f)} \sum_{P \in \text{supp}(r')}    r'(P) \\
&  = 2^{2\R_\varepsilon(f)}.
\end{align}
The first inequality holds since $|I_S| \leq  \R_\varepsilon(f)$, and the second inequality uses the fact that
$|P| \leq 2^{\R_\varepsilon(f)}$. Setting $\varepsilon=0$ gives $\mathrm{PPRT}_0(f) \leq \D(f)$.
\end{proof}

The following theorem summarizes the known relations between the introduced complexity measures.

\begin{figure}[htb]
    \centering
   \begin{tikzpicture}[x=1cm,y=0.9cm]

     \node (Df) at(2,6) [align=center]{$D(f)$};
     \node (Dscf) at(2,4) [align=center]{$\D(f)$};
     \node (Rf) at(0,5) [align=center]{$R_\varepsilon(f)$};
     \node (Rscf) at(0,3) [align=center]{$\R_\varepsilon(f)$};
     \node (pprt) at(0,1) [align=center]{$\mathrm{PPRT}_\varepsilon(f)$};
     \node (pprt0) at(2,2) [align=center]{$\mathrm{PPRT}_0(f)$};
     \node (R0f) at(4,5) [align=center]{$R_0(f)$};
     \node (R0scf) at(4,3) [align=center]{$\R_0(f)$};

     \path[->] (Rf) edge (Df);
     \path[->] (R0f) edge (Df);
     \path[->] (Dscf) edge (Df);
     \path[->] (Rscf) edge (Rf);
     \path[->] (R0scf) edge (R0f);
     \path[->] (Rscf) edge (Dscf);
     \path[->] (R0scf) edge (Dscf);
     \path[->] (pprt) edge (Rscf);
     \path[->] (pprt0) edge (Dscf);
     \path[->] (pprt) edge (pprt0);

   \end{tikzpicture}
    \caption{Relationships between the complexity measures introduced. An arrow from $X$ to $Y$ represents $X \leq Y$. For example, $\D(f) \rightarrow D(f)$ means $\D(f) \leq D(f)$.\label{fig:rel}}
\end{figure}
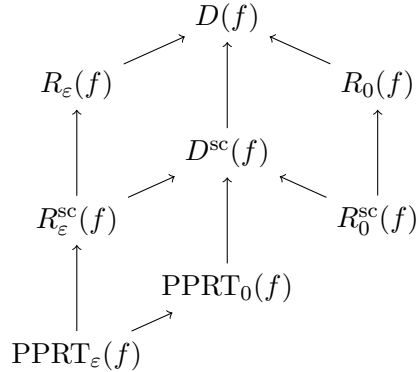

\begin{restatable}{theorem}{summary}
\label{thm:summary}
For any Boolean function $f:\{0,1\}^n \to \{0,1\}$ and for all $\varepsilon > 0$, the relations indicated in \fig{rel} hold.
\end{restatable}

\begin{proof}
The upper three vertical arrows represent the relations between query complexity and subcube partition complexity established in \prop{query-sc}. The remaining vertical arrows represent the relations between the public-coin partition bounds and subcube partition complexity established in \prop{pprt-rsc}. The other inequalities are immediate and follow from their definitions.
\end{proof}

\section{Iterated quaternary majority function}
\label{sec:it-fmaj}

We now introduce the function we use to separate randomized query complexity from  subcube partition complexity and establish some of its properties.


Let $\MAJ$ denote the Boolean majority function of its input bits when the number of bits is odd. The quaternary majority function
$\FMAJ : \{0,1\}^4 \rightarrow \{0,1\}$ is defined by $\FMAJ(x_1,x_2,x_3,x_4)=x_1(x_2 \vee x_3 \vee x_4)\vee x_2 x_3 x_4$. This function was introduced in \cite{Sav02}. We call it $\FMAJ$, because the output of the function is the majority of its input bits, with
the first variable  breaking equality in its favor. In other words, the first variable has two votes, while the others have one, that is
$\FMAJ(x_1,x_2,x_3,x_4) = \MAJ(x_1,x_1,x_2,x_3,x_4)$.
This function has previously been used to separate deterministic decision tree size from deterministic subcube partition size \cite{Sav02}. We use this function because its subcube partition complexity is smaller than its query complexity.

\begin{restatable}{proposition}{fmajdet}
\label{prop:fmaj-det}
We have $\D(\FMAJ) = 3$ and $D(\FMAJ)=4$.
\end{restatable}

\begin{proof}
Observe that, for any choice of $w \in \{0,1\}$, we have that
\begin{align*}
 \FMAJ(0,0,1,w) = \FMAJ(0,w,0,1) & = \FMAJ(0,1,w,0) = \FMAJ(w,0,0,0) = 0 
   \\
   \text{ and that } \\
 \FMAJ(1,1,0,w) = \FMAJ(1,w,1,0) & = \FMAJ(1,0,w,1) = \FMAJ(w,1,1,1)=1.
\end{align*} 
The subcubes generated by these 8 partial assignments are disjoint and of size two, forming a partition of $\{0,1\}^4$.  Thus, with the right Boolean values, they form a deterministic subcube partition that computes $\FMAJ$.
Since all partial assignments have length 3, $\D(\FMAJ) \leq 3$. 
Although we do not use the inequality $\D(\FMAJ) \geq 3$ in our results, this can be verified by enumerating all deterministic subcube partitions with complexity 2. Furthermore, $D(\FMAJ)\leq 4$ since any function can be computed by querying all input bits. $D(\FMAJ)\geq 4$ can be shown either by enumerating all decision trees that make 3 queries or by using the lower bound in the next section.
\end{proof}

While our results only require us to show lower bounds on the randomized query complexity of $\FMAJ$, we want to mention that the randomized query complexity of $\FMAJ$ is indeed smaller than its deterministic query complexity.

\begin{restatable}{proposition}{fmajrand}
\label{prop:fmaj-rand}
For the $\FMAJ$ function, $R_0(\FMAJ) \leq 13/4 = 3.25$.
\end{restatable}

\begin{proof}
The randomized algorithm achieving this complexity is simple: with probability 1/4, the algorithm queries the first variable and then it checks if the other variables all have the opposite value; with probability 3/4, it checks if the last three variables have all the same value and, if not, it queries the first variable.
\end{proof}


Since the $\FMAJ$ function separates deterministic subcube complexity from deterministic query complexity, a natural candidate for a function family that separates these measures is the iterated quaternary majority function, $\FMAJ_h$, defined recursively on $4^h$ variables, for $h \geq 0$. In the base case, $\FMAJ_0$ is the identity function on one bit. For $h>0$, we define $\FMAJ_h = \FMAJ \circ \FMAJ_{h-1}$.
In other words, for $h >0$, let $x$ be an input of length $4^h$, and for $i \in \{1,2,3,4\}$, let $x^{(i)}$ denote the $i^{th}$ quarter of $x$,
that is $|x^{(i)}| = 4^{h-1}$ and $x = x^{(1)} x^{(2)} x^{(3)} x^{(4)}$.
Then, we have that
$\FMAJ_h(x) = \FMAJ ( \FMAJ_{h-1} (x^{(1)})$, $ \FMAJ_{h-1} (x^{(2)})$, $ \FMAJ_{h-1} (x^{(3)})$, $ \FMAJ_{h-1} (x^{(4)}) ).$

The function $\FMAJ_h$ inherits several properties from $\FMAJ$. It has low deterministic subcube complexity, but high deterministic query complexity:

\begin{proposition}
\label{prop:dsc-h}
For all $h \geq 0$, $\D(\FMAJ_h) \leq 3^h$ and $D(\FMAJ_h) = 4^h$.
\end{proposition}

\begin{proof}
For $h=0$, the statement is trivial and for $h=1$, the statement is \prop{fmaj-det}. 
\prop{det-comp} and \prop{dsc-comp} used recursively imply the result.
\end{proof}

We now introduce terminology that we use to refer to this function. We view $\FMAJ_h$ as defined by the read-once formula on the complete quaternary tree $\rT_h$ of height $h$ in which every internal node is a $\FMAJ$ gate. 
We identify the leaves of $\rT_h$ from left to right with the integers $1, \ldots, 4^h$. For an
input $x \in \{0,1\}^{4^h}$, the bit $x_i$ defines the {\em value} of the leaf $i$. We then evaluate recursively the values of the internal nodes. The value of
the root is $\FMAJ_h(x)$. 
For every internal node $v$ in $\rT_h,$ 
we denote its children by $v_1, v_2, v_3$ and $v_4$, from left to right.
For any node $v$ in $\rT_h$, 
let $Z(v)$ denote the set of variables associated with the leaves
in the subtree rooted at $v$.
We say that a
node $v$ is at level $\ell$ in $\rT_h$ if the distance between $v$ and the
leaves is $\ell$. The root is therefore at level $h$, and the leaves are at
level $0$. For $0 \leq \ell \leq h$, the set nodes at level $\ell$ is denoted by $\rT_h(\ell)$.

\section{Randomized query complexity of \texorpdfstring{$\FMAJ_h$}{4MAJ}}
\label{sec:lb-fmaj}

In this section, we prove our main technical result, a lower bound on the randomized query complexity of $\FMAJ_h$. We prove this by using distributional complexity, that is by using the inequality in \prop{yao}. First, we define a ``hard distribution'' $d_h$ for which we will show that $\Delta_{\varepsilon}^{d_h}(\FMAJ_h) \geq (1-2\varepsilon)(16/5)^h$, which implies our main result (\thm{main}).


\subsection{The hard distribution}

Intuitively, the distribution we use in our lower bound has to be one on which it is difficult to compute $\FMAJ_h$. We start by defining a hard distribution for $\FMAJ$ and extend it to $\FMAJ_h$ in the natural way: by composing it with itself.

The {\em hard distribution $d$} on inputs of length $4$ is defined from $d^0$ and $d^1$, the respective hard distributions for $0$-inputs and $1$-inputs of length $4$, by setting $d(x) = \frac{1}{2} d^b(x)$ when $\FMAJ(x) = b$. 
We define $d^0$ as 
\begin{align}
d^0(1000) &= \frac{2}{5}, \quad d^0(0011) = d^0(0101) = d^0(0110) = \frac{1}{6},\nonumber \\
d^0(0001) &= d^0(0010) = d^0(0100) = \frac{1}{30}, \textrm{\; and \;} d^0(0000) = 0.
\end{align}
The definition of $d^1$ is analogous, or can be defined by $d^1(x_1, x_2, x_3, x_4) = d^0 (1-x_1, 1-x_2, 1-x_3, 1-x_4)$. Given that the function $\FMAJ$ is symmetric in $x_2$, $x_3$, and $x_4$, there are only 4 equivalence classes of $0$-inputs, to which 
we have assigned probability masses $2/5, 1/2, 1/10$, and $0$, and then distributed the probabilities uniformly
inside each class.
The probabilities were chosen to make the recurrence relations in \lem{recursion-J} and \lem{recursion-K} work, while putting more weight on the intuitively difficult inputs. For example $x=0000$ seems like an easy input since all inputs that are Hamming distance 1 from it are also $0$-inputs, and thus reading any $3$ bits of this input is sufficient to compute the function.
In \lem{recursion-K} we will give an equivalent characterisation of the hard distribution which is more directly related to the 
recurrence relations in the lemmas.

From this distribution we recursively define, for $h \geq 0$, the {\em hard distribution $d_h$} on inputs of length $4^h$. In the base case,
$d_0(0) = d_0(1) = \frac{1}{2}$. For $h >0$, 
as for $d$, the distribution $d_h$ is defined from $d^0_h$ and $d^1_h$, the respective hard distributions for 0-inputs and 1-inputs
of length $4^{h}$, by setting $d_h(x) = \frac{1}{2} d^b_h(x)$ when $\FMAJ(x) = b$. 
Let $x = x^{(1)} x^{(2)} x^{(3)} x^{(4)}$ be a $b$-input, where $x^{(i)}$ is a $b_i$-input of length $4^{h-1}$, for
$i \in \{1,2,3,4\}$. Then, $d_h^b(x) = d^b(b_1b_2b_3b_4) \cdot \Pi_{i=1}^4 d_{h-1}^{b_i}(x^{(i)})$.
It is easily seen that according to $d_h$, for  each
node $v$ in $\rT_h$,   
if the value of $v$ is $b$, then the children of $v$ have values distributed according to $d^b$.
With the additional constraints that the root has uniform distribution over $\{0, 1\}$, this actually makes an alternative definition of $d_h$.


We will also require the notion of a minority path in our proof. For a given input, a minority path is a path from the root to a leaf in which each node has a value different from its parent's value. (Recall that the value of a node is the function $\FMAJ$ evaluated on the values of its children.) For example, for the $\FMAJ$ function, on  input $1000$ the unique minority path 
is the edge from the root to the first variable, whereas on input 1001 there are two minority paths from the root to the second and third variable.
In general, since there may be multiple such paths, the minority path is defined to be a random variable over all root--leaf paths.
Formally, for every input $x \in \{0,1\}^{4^h}$, we define the \emph{minority path} $M(x)$ as a random variable over all root--leaf paths in $\rT_h$ as follows. First, the root is always in $M(x)$. Then, for any node $v$ in $M(x)$, if there is a unique child $w$ of $v$ with value different from that of $v$, then $w \in M(x)$. Otherwise, there are exactly two children with different values, and we put each of them in $M(x)$ with probability $\frac{1}{2}$. 
Note that with this definition, if $x$ is chosen from the hard distribution $d_h$, conditioned on the node $v$ being in $M(x)$, the first child $v_1$ is in the minority path with probability $\frac{2}{5}$, and the child $v_i$ is in the minority path with probability $\frac{1}{5}$, for $i \in \{2,3,4\}$.

\subsection{Complexity of \texorpdfstring{$\FMAJ_h$}{4MAJ} under the hard distribution}

We can now lower bound the distributional complexity of $\FMAJ_h$ under the hard distribution.

\begin{restatable}{theorem}{dlower}
\label{thm:distributional-lower}
For all $\varepsilon \geq 0$ and $h \geq 0$, we have 
$\Delta_{\varepsilon}^{d_h}(\FMAJ_h) \geq  (1-2\varepsilon)({16}/{5})^h.$
\end{restatable}

To show this, we need to define some quantities.
For a deterministic decision tree algorithm $A$ computing $\FMAJ_h$, let
$L_A(x)$ denote the set of variables queried by $A$ on input $x$.
Let $B$ be a randomized decision tree algorithm that computes $4$-MAJ$_h$ with error $\varepsilon$, and let $b$ be its probability distribution over deterministic algorithms. For any two (not necessarily distinct) nodes of $\rT_h$, $u$ and $v$, we define the function $E_B(v,u)$ as
$E_B(v,u) \defeq \mathbb{E}\big[|Z(v) \cap L_A(x)| \big| u \in M(x)\big],$
where the expectation is taken over $b, d_h$ and the randomness in $M(x)$. In words, $E_B(v,u)$ is the expected number of queries below the node $v$
over the randomness of $B$, the hard distribution and the randomness for the choice of the minority path, 
under the condition that $u$ is in the minority path.
For $ 0 \leq \ell \leq h$, we also define the functions $J_B^\varepsilon(h,\ell)$, $K_B^\varepsilon(h,\ell)$, $J^\varepsilon(h,\ell)$, and $K^\varepsilon(h,\ell)$ by
\begin{align}
J_B^\varepsilon(h,\ell) &\defeq  \sum_{v \in \rT_{h}(\ell)} E_B(v,v), \\
K_B^\varepsilon(h,\ell) &\defeq    \sum_{v \in \rT_{h}(\ell)}
\left(\frac{2}{5} \sum_{i=2}^4 E_B(v_i, v_1) +\frac{1}{5} \sum_{j=2}^4  \sum_{i \neq j}  
E_B(v_i, v_j) \right),
\end{align}
\begin{equation}
J^\varepsilon(h,\ell) \defeq \min_{B \in \mathcal{R}_{\varepsilon}(\FMAJ_h)} J_B^\varepsilon(h,\ell) \quad \text{and} \quad 
K^\varepsilon(h,\ell) \defeq \min_{B \in \mathcal{R}_{\varepsilon}(\FMAJ_h)} K_B^\varepsilon(h,\ell).
\end{equation}
Observe that $J^\varepsilon(h,h) =  \min_{B \in \mathcal{R}_{\varepsilon}(\FMAJ_h)} \mathbb{E}[C(A,x)]
\leq \Delta_{\varepsilon}^{d_h}(\FMAJ_h)$.

The proof of \thm{distributional-lower} essentially follows from the following two lemmas.

\begin{restatable}{lemma}{recJ}
\label{lem:recursion-J}
For all $0<l\leq h$, we have that
$J^\varepsilon(h,\ell) \geq K^\varepsilon(h,\ell)+\frac{1}{5}J^\varepsilon(h,\ell-1).$
\end{restatable}

\begin{proof}
This proof mainly involves expanding the quantity $E_B(v,v)$ in terms of $E_B(v_i,v_j)$, where $v_1, v_2, v_3$, and $v_4$ are the children of $v$.
Since, for every node $v$, the set of leaves below $v$ is the disjoint union of the sets of leaves below 
its children, for every $B$ we have that
\begin{equation}
J_B^\varepsilon(h,\ell) =
\sum_{v \in \rT_{h}(\ell)} \sum_{i=1}^4 E_B(v_i,v).
\end{equation}
By conditioning on the minority child of $v$, we get that
\begin{equation}
J_B^\varepsilon(h,\ell) =
\sum_{v \in \rT_{h}(l)} \sum_{i=1}^4 \sum_{j=1}^4
E_B(v_i, v_j) \Pr [v_j \in M(x) | v \in M(x)] ~.
\end{equation}
As mentioned before, if $x$ is chosen according to the distribution $d_h$, if $v \in M(x)$, then $v_1 \in M(x)$ with probability $\frac{2}{5}$ and $v_i \in M(x)$ with probability $\frac{1}{5}$, for $i \in \{2,3,4\}$. Substituting these values we get
\begin{equation}
J_B^\varepsilon(h,\ell) = K_B^\varepsilon(h,\ell)+\frac{1}{5}J_B^\varepsilon(h,\ell-1) + \frac{1}{5} E_B(v_1, v_1).
\end{equation}
Discarding the last term on the right hand side, which is always non-negative, and 
taking the minimum over $B$ for all remaining terms gives the result.
\end{proof}

Having established this, we need to relate $K^\varepsilon(h,\ell)$ with $J^\varepsilon(h-1,\ell-1)$. Informally, given a randomized algorithm that performs well on $\FMAJ_h$ at depth $\ell$, we construct another algorithm that performs well on $\FMAJ_{h-1}$ at depth $\ell-1$.

\begin{restatable}{lemma}{recK}
\label{lem:recursion-K}
For all $0 < \ell \leq h$, we have that $K^\varepsilon(h,\ell) \geq 3J^\varepsilon(h-1,\ell-1)$.
\end{restatable}

\begin{proof}
For any $B \in \mathcal{R}_\varepsilon(\FMAJ_h)$, we will construct $B' \in \mathcal{R}_\varepsilon(\FMAJ_{h-1})$ such that
\begin{equation}
\label{eq:KJ}
\frac{1}{3}K_B^\varepsilon(h,\ell) = J_{B'}^\varepsilon(h-1,\ell-1).
\end{equation}
Taking the minimum over all $B \in \mathcal{R}_\varepsilon(\FMAJ_h)$ implies the statement.

We start by giving a high level description of our construction of $B'$ from $B$.
First $B'$ will choose a random injective mapping from  $\{x_1, \ldots , x_{4^{h-1}}\}$ to $\{x_1, \ldots , x_{4^{h}}\}$, identifying each variable of $\rT_{h-1}$ with some variable of $\rT_h$. Then, it will choose a random restriction for the remaining variables of $\rT_h$. Note that these choices are not be made uniformly. Let $B_r$ denote  the algorithm for $4^{h-1}$ variables defined by $B$ after the identification and the restriction according to randomness $r$. $B'$ then simply executes $B_r$. Our embedding of the smaller instance into the larger instance is done in  a way that preserves the output.

We now describe the random identification and restriction in detail.
First, observe that there is a natural correspondence between the nodes of $\rT_{h-1}(\ell - 1)$ and $\rT_h(\ell)$ (since they are of the same size):
we simply map the $i$th node of $\rT_{h-1}(\ell - 1)$ from the left to the $i$th node of $\rT_{h}(\ell)$ from the left.
For every node $u \in \rT_{h-1}(\ell - 1)$, let $v \in \rT_h(\ell)$ be its corresponding node. The algorithm $B'$ makes the following independent random choices.
To generate the random identification, $B'$ randomly chooses a child $w$ of $v$, where 
$w = v_1$ with probability $\frac{1}{5}$, and $w= v_i$ with probability $\frac{4}{15}$, for $i \in \{2,3,4\}$.
Then, the variables of $Z(u)$ and the variables of $Z(w)$ are identified naturally, again from left to right.

For generating the random restriction, $B'$ first generates random values for the three siblings of $w$.
If $w=v_1$, then it chooses for $(v_2, v_3, v_4)$ one of the
six strings from $\{001, 010, 100, 110, 101, 011\}$ uniformly at random. If $w \in \{v_2, v_3, v_4\}$, it chooses for $v_1$,
a uniformly random value from $\{0,1\}$, and for the remaining two siblings, it picks the opposite value.
From this, the restriction is generated as follows: for each sibling $w'$ of $w$ with value $b \in \{0,1\}$, 
a random string of length $4^{\ell - 1}$ is generated according to $d^b_{\ell -1}$, and the variables in $Z(w')$ receive
the values of this string. This finishes the description of $B'$.

We now show that $B' \in \mathcal{R}_\varepsilon(\FMAJ_{h-1})$.
Because of the identification of the variables of $Z(u)$ and $Z(w)$, for every $x \in \{0,1\}^{4^{h-1}}$, the value of $u$ coincides
with the value of $w$. The random values chosen for $w$ are such that whatever value $w$ gets, it is always a majority child of $v$.
Therefore, for every input $x$, and for every randomness $r$, the value of $u$ is the same as the value of $v$. This implies that
for every $x$ and every randomness $r$, the value of the roots of $\rT_{h-1}(\ell - 1)$ and $\rT_h(\ell)$ are the same. 
Since $B$ is an algorithm which computes $\FMAJ_h$ with error at most $\varepsilon$, this means that $B_r$
is an algorithm which computes $\FMAJ_{h-1}$ with error at most $\varepsilon$, for every randomness $r$. From this, it
follows that $B' \in \mathcal{R}_\varepsilon(\FMAJ_{h-1})$.

Finally we prove 
the equality in \eq{KJ}. 
For this, the main observation (which can be checked by direct calculation) is
that when $w$ gets a random Boolean value, the distribution of
values generated by $B'$ on the children of $v$ is exactly the hard distribution $d$. 
Therefore, $E_{B'}(u,u) = E_B(w,v)$. Consequently, we have that
\begin{align}
J_{B'}^\varepsilon(h-&1,\ell-1) = \sum_{v \in \rT_h(\ell)} E_B(w,v)   = \sum_{v \in \rT_h(\ell)} \sum_{i=1}^4 E_B(v_i,v) \Pr [w=v_i | v \in M(x)] \nonumber \\
& = \sum_{v \in \rT_h(\ell)} \sum_{i=1}^4 \sum_{j=1}^4 E_B(v_i, v_j) \Pr [w=v_i ] \Pr [v_j \in M(x) | w=v_i, v \in M(x)] \nonumber \\
& =  \frac{1}{3}K_B^\varepsilon(h,\ell).
\end{align}
The third equality holds since the choice of $w$ is independent from the fact that $v$ is in the minority path. For the
last equality, we used that the conditional probabilities evaluate to the following values:
\begin{align*}
\Pr [v_j \in M(x) | w=v_j, v \in M(x)] & = 0,   \qquad \text{for} ~ j \in \{1,2,3,4\}; \\
\Pr [v_j \in M(x) | w=v_1, v \in M(x)] & = \frac{1}{3},  \qquad  \text{for} ~ j \neq 1; \\
\Pr [v_1 \in M(x) | w=v_i, v \in M(x)] & = \frac{1}{2},   \qquad \text{for} ~ i \neq 1; \\
\Pr [v_j \in M(x) | w=v_i, v \in M(x)] & = \frac{1}{4}, \qquad  \text{for} ~ i,j \in \{2,3,4\} \text{ and } i \neq j. \qedhere
\end{align*}
\end{proof}

We can now return to proving \thm{distributional-lower}.

\begin{proof}[Proof of \protect{\thm{distributional-lower}}]
We claim that, for all $0 \leq \ell \leq h,$ we have that
\begin{equation}
J^\varepsilon(h,\ell) \geq (1 - 2 \varepsilon) ({16}/{5})^{\ell}.
\end{equation}
The proof is done by induction on $\ell$. For the base case $\ell =0$, let $B \in \mathcal{R}_{\varepsilon}(\FMAJ_h)$. Then, we have that
\begin{equation}
J_B^\varepsilon(h,0) = \sum_{v \in \rT_{h}(0)} \Pr [B~ \mbox{{\rm  queries }} v \big| v \in M(x) ].
\end{equation}
Observe that any randomized decision tree algorithm 
computing a nonconstant function
with error at most $\varepsilon$ must make at least one query  with probability at least $1-2\varepsilon$, since otherwise it would output $0$ or $1$  with probability greater than $\varepsilon$, and thus on some input would err too much. 
Let therefore $A$ be a deterministic algorithm from the support of $B$ 
which makes at least one query. Then
\begin{equation}
\sum_{v \in \rT_{h}(0)} \Pr [A~ \mbox{{\rm  queries }} v \big| v \in M(x) ] \geq
\sum_{v \in \rT_{h}(0)} \Pr [A~ \mbox{{\rm first query  is }} v \big| v \in M(x) ] = 1,
\end{equation}
since in the summation the term corresponding to the first query of $A$ is 1, whereas all other terms are 0.
Thus, $J(h,0) \geq 1-2\varepsilon$ for all $h \geq 0$.

Now let $\ell >0$, and assume the statement holds for $\ell -1$. For $h \geq \ell$,
using  \lem{recursion-J} and \lem{recursion-K}, we get that
$J^\varepsilon(h,\ell) \geq 3J^\varepsilon(h-1,\ell-1)+\frac{1}{5}J^\varepsilon(h,\ell-1)$. 
Therefore, by the induction hypothesis, we have that
\begin{equation}
J^\varepsilon(h,\ell) \geq 3(1-2\varepsilon)\bigg(\frac{16}{5}\bigg)^{\ell-1}
+\frac{1}{5}(1-2\varepsilon)\bigg(\frac{16}{5}\bigg)^{\ell-1} \\
= (1-2\varepsilon)\bigg(\frac{16}{5}\bigg)^l~.
\end{equation}
The theorem follows when we set $h = \ell$ by noting that  $J^\varepsilon(h,h) \leq \Delta_{\varepsilon}^{d_h}(\FMAJ_h)$.
\end{proof}

Combining \prop{dsc-h} and \thm{distributional-lower} gives us our main result, an asymptotic separation between deterministic subcube partition complexity and randomized query complexity:

\main*

We can also immediately deduce that the $\FMAJ_h$ function positively answers both \ques{fkw} and \ques{jain}.

\begin{corollary}\label{c-friedgut}
We have that $\R_0(\FMAJ_h) = o(R_0(\FMAJ_h)).$
\end{corollary}

\begin{corollary}\label{c-jain}
For $0 \leq \varepsilon \leq 1/3$, we have that $\text{\emph{PPRT}}_\varepsilon(\FMAJ_h) = o(R_\varepsilon(\FMAJ_h)).$
\end{corollary}

\section{Discussion and open problems}
\label{sec:disc}

Our main result is actually stronger than stated. In addition to the zero-error and $\varepsilon$-error randomized query complexities we defined, we can also define $\varepsilon$-error expected randomized complexity. In this model, we only charge for the expected number of queries made by the randomized algorithm, like in the zero-error case, but we also allow the algorithm to err. Formally, the \emph{$\varepsilon$-error expected randomized query complexity of $f$} is 
$R_\varepsilon^{\e}(f)= \min_{B \in \mathcal{R}_{\varepsilon}(f)}\max_{x} C(B,x))$. 
Observe that since this generalizes zero-error randomized query complexity, $R^\e_0(f) = R_0(f)$, and it 
is immediate that, for all $\varepsilon \geq 0$, we have that $R^{\e}_{\varepsilon}(f) \leq R_{\varepsilon}(f) \leq D(f)$.

Randomized query complexity is usually defined in the worst case~\cite{BW02}, that is as $R_\varepsilon(f)$ instead of $R^{\e}_{\varepsilon}(f)$. The main reason for not dealing with these measures separately is that worst case and expected randomized complexities are closely related. We have already observed that (obviously), in expectation, one can not make more queries than in the worst case. On the other hand, if for some constant $\eta >0$ we let the randomized algorithm that achieves $R_\varepsilon^{\e}(f)$ make $\frac{1}{2\eta} R_\varepsilon^{\e}(f)$ queries, and give a random answer in case the computation is not finished, we get an algorithm of error $\varepsilon + \eta$ which never makes more than $\frac{1}{2\eta} R_\varepsilon^{\e}(f)$ queries. Therefore, for all $\varepsilon \geq 0$ and $\eta > 0$, we have that $R_{\varepsilon + \eta}(f) \leq \frac{1}{2\eta} R_\varepsilon^{\e}(f)$.

The result we show actually lower bounds $R^{\e}_{\varepsilon}(f)$ as well. Thus, a stronger version of our result is the following: For all $\varepsilon \geq 0$,  $R_\varepsilon^{\e}(\FMAJ_h) \geq (1-2\varepsilon)(3.2)^h$.

We end with some open problems. It would be interesting to exactly pin down the randomized query complexity of $\FMAJ_h$. For example we know that $R_0(\FMAJ_h) \geq 3.2^h$ and $R_0(\FMAJ_h) \leq 3.25^h$. The best separation between subcube partition complexity and query complexity remains open, even in the deterministic case. For example, we know that $\D(f) \leq D(f)$ and $D(f) \leq (\D(f))^2$, so they are at most quadratically different. The $\FMAJ_h$ function shows that there exists a function for which $D(f) \geq \D(f)^{\log_3 4} \geq (\D(f))^{1.26}$. Can this separation or the quadratic upper bound be improved?

Finally it would be interesting to know if the partition bounds also lower bound expected randomized query complexity, and in particular whether the zero-error partition bound lower bounds zero-error randomized query complexity. 

\section*{Acknowledgments}

The research is partially funded by the Singapore Ministry of Education and the
National Research Foundation, also through the Tier 3 Grant ``Random numbers
from quantum processes,'' MOE2012-T3-1-009; the European Commission
IST STREP project Quantum Algorithms (QALGO) 600700; the French ANR
Blanc program under contract ANR-12-BS02-005 (RDAM project); and 
the ARO grant Contract Number W911NF-12-1-0486. This preprint is MIT-CTP \#4663.

\bibliographystyle{alpha}
\bibliography{partition}

\end{document}